\documentclass[11pt]{article}
\usepackage[letterpaper,margin=0.8in]{geometry}
\usepackage{amsmath, amssymb, amsthm, thmtools, amsfonts}

\usepackage{graphicx}

\usepackage{tikz}
\usetikzlibrary{positioning,decorations.pathreplacing}

\usepackage{cite}
\usepackage{appendix}
\usepackage{graphicx}
\usepackage{color}
\usepackage{algorithm}
\usepackage[noend]{algpseudocode}
\usepackage{epstopdf}
\usepackage[textsize=tiny]{todonotes}

\usepackage{framed}
\usepackage[framemethod=tikz]{mdframed}
\usepackage[bottom]{footmisc}
\usepackage[shortlabels]{enumitem}
\setitemize{noitemsep,topsep=3pt,parsep=3pt,partopsep=3pt}
\usepackage[font=small]{caption}
\usepackage{caption}
\usepackage{subcaption}
\usepackage{xspace}

\newtheorem{theorem}{Theorem}[section]
\newtheorem{lemma}[theorem]{Lemma}
\newtheorem{meta-theorem}[theorem]{Meta-Theorem}
\newtheorem{claim}[theorem]{Claim}
\newtheorem{remark}[theorem]{Remark}
\newtheorem{corollary}[theorem]{Corollary}

\definecolor{darkgreen}{rgb}{0,0.5,0}
\usepackage{hyperref}
\hypersetup{
    unicode=false,          
    colorlinks=true,        
    linkcolor=red,          
    citecolor=darkgreen,        
    filecolor=magenta,      
    urlcolor=cyan           
}
\usepackage[capitalize, nameinlink]{cleveref}
\crefname{theorem}{Theorem}{Theorems}
\Crefname{lemma}{Lemma}{Lemmas}
\Crefname{claim}{Claim}{Claims}

\algnewcommand\algorithmicswitch{\textbf{switch}}
\algnewcommand\algorithmiccase{\textbf{case}}

\algdef{SE}[SWITCH]{Switch}{EndSwitch}[1]{\algorithmicswitch\ #1\ \algorithmicdo}{\algorithmicend\ \algorithmicswitch}%
\algdef{SE}[CASE]{Case}{EndCase}[1]{\algorithmiccase\ #1}{\algorithmicend\ \algorithmiccase}%
\algtext*{EndSwitch}%
\algtext*{EndCase}%
\newcommand{\eps}{\varepsilon}
\newcommand{\poly}{\textrm{poly}}

\newcommand{\mpc}{\mathsf{MPC}}
\newcommand{\local}{\mathsf{LOCAL}}

\newcommand{\mis}{\mathsf{MIS}}
\newcommand{\parent}{\textrm{parent}}

\newcommand{\bigO}{O}

\renewcommand{\paragraph}[1]{\vspace{0.15cm}\noindent {\bf #1}:}
\newcommand{\para}[1]{\vspace{0.15cm}\noindent {\bf #1}}

\begin{document}

\date{}

\author{
Sebastian Brandt\\
  \small ETH Zurich \\
  \small brandts@ethz.ch
	\and
	 Manuela Fischer\\
  \small ETH Zurich \\
  \small manuela.fischer@inf.ethz.ch
	\and
	 Jara Uitto\\
  \small ETH Zurich \& Uni. Freiburg\\
  \small jara.uitto@inf.ethz.ch
 }

\setcounter{page}{0}

\title{Breaking the Linear-Memory Barrier in $\mpc$:  \\ Fast $\mis$ on Trees with Strongly Sublinear Memory}

\maketitle

\thispagestyle{empty}

\begin{abstract}
Recently, studying fundamental graph problems 
 in the \emph{Massively Parallel Computation ($\mpc$)} framework, inspired by the \emph{MapReduce} paradigm, has gained a lot of attention.
	An assumption common to a vast majority of approaches 
	is to allow $\widetilde{\Omega}(n)$ memory per machine, where $n$ is the number of nodes in the graph and $\widetilde{\Omega}$ hides polylogarithmic factors.
	However, as pointed out by Karloff et al.~[SODA'10] and Czumaj et al.~[STOC'18], 
	it 
	might be unrealistic for 
	a single machine 
	to have linear or only slightly sublinear memory. 
	
	
	In this paper, we thus study a more practical variant of the $\mpc$ model which only requires substantially sublinear or even subpolynomial memory per machine. In contrast to the linear-memory $\mpc$ model and also to streaming algorithms, in this low-memory $\mpc$ setting, a single machine will only see a small number of nodes in the graph. 
	%
	We introduce a new and strikingly simple technique to cope with this imposed locality. 

In particular, we show that the \emph{Maximal Independent Set ($\mis$)} problem can be solved efficiently, that is, in $\bigO(\log^3 \log n)$ rounds, when the input graph is a tree. This constitutes an almost exponential speed-up over the low-memory $\mpc$ algorithm in $\widetilde{O}(\sqrt{\log n})$-algorithm in a concurrent work by Ghaffari and Uitto [SODA'19] and substantially reduces the local memory from $\widetilde{\Omega}(n)$ required by the recent $\bigO(\log \log n)$-round $\mis$ algorithm of Ghaffari et al.~[PODC'18] to $n^{\eps}$ for any $\eps>0$, without incurring a significant loss in the round complexity. Moreover, it demonstrates how to make use of the all-to-all communication in the MPC model to almost exponentially improve on the corresponding bound in the $\mathsf{LOCAL}$ and $\mathsf{PRAM}$ models by Lenzen and Wattenhofer [PODC'11].	
\end{abstract}
\newpage

\section{Introduction}
\paragraph{Parallel Computation Paradigms for Massive Data}
When confronted with huge data sets, purely sequential approaches become untenably inefficient. To address this issue, several parallel computation frameworks specially tailored for processing large scale data have been introduced. 
%
Inspired by the MapReduce paradigm \cite{dean2008mapreduce}, Karloff, Suri, and Vassilvitskii \cite{karloff2010model} proposed the \emph{Massively Parallel Computation ($\mpc$)} model, 
which was later refined in many works \cite{goodrich2011sorting,beame2014skew,Andoni2014,beame2017communication,czumaj2017round}.

\paragraph{Massively Parallel Computation Model}
In the $\mpc$ model, an input instance of size $N$ is distributed across $M$ machines with local memory of size $S$ each. 
The computation proceeds in rounds, each round consisting of \emph{local computation} at the machines interleaved with \emph{global communication} (also called \emph{shuffling}, adopting the MapReduce terminology) between the machines. 

In the shuffling step, every machine is allowed to send as many messages to as many machines as it wants, as long as for every machine the total size of sent and received messages does not exceed its local memory capacity. The quantity of main interest is the round complexity: the number of rounds needed until the problem is solved, that is, until every machine outputs its part of the solution. This measure constitutes a good estimate for the actual running time, as local computation is presumed to be negligible compared to the cost-intensive shuffling, which requires a massive amount of data to be transferred between machines. 

\paragraph{Sublinear Memory Constraint}  
Note that $S \geq N$ leads to a degenerate case that allows for a trivial solution. Indeed, as the data fits into the local memory of a single machine, the input can be loaded there, and a solution can be computed locally. 
Due to the targeted application of $\mpc$ in the presence of massive data sets, thus large $N$, it is often crucial that $S$ is not only smaller than $N$ but actually substantially sublinear in $N$. 
The total memory $M\cdot S$ in the system has to be at least $N$, so that the input actually fits onto the machines, but ideally not much larger. 
Summarized, one requires $S=\widetilde{O}\left(N^{\eps}\right)$ memory on each of the $M=\widetilde{O}\left(N^{1-\eps'}\right)$ machines, for $0<\eps'\leq \eps<1$. 

\paragraph{Sublinear Memory for Graph Problems}
Basically all known $\mpc$ techniques for graph problems need essentially linear in $n$---for instance, $\widetilde{\Omega}(n)$ or mildly sublinear like $n^{1-o(1)}$---memory per machine, where $n$ is the number of nodes in the input graph\footnote{In the context of graph problems, it is typical to assume that all incident edges of a node are stored on the same machine, resulting in two copies of an edge, one for each endpoint. We refer to \cite[Section 1.1]{pandurangan2016fast} for a thorough discussion. Also see the remark at the end of this section.}. We refer to \cite{behnezhad2018brief} for a brief discussion of this assumption. Note that for sparse graphs with $N=\widetilde{O}(n)$ edges, this 
violates the sublinear memory constraint, getting close to the degenerate regime. This issue has been artificially circumvented by explicitly restricting the attention to dense graphs with $N=\widetilde{\Omega}(n^{1+\eps})$ edges, as to ensure sublinearity in $N$ while still not having to relinquish the nice property that (essentially) all nodes fit into the memory of a single machine \cite{karloff2010model}. 

Besides being a stretch of the definition, this additionally imposed condition of denseness of the input graph does not seem to be realistic. In fact, as recently also pointed out by \cite{czumaj2017round}, most practical large graphs are sparse. 
For instance in the Internet, most of the nodes have a small degree.  
Even for dense graphs, where in theory the sublinear memory constraint is met, practicability of the parameter range does not need to be ensured; 
linear or slightly sublinear in $n$ might be prohibitively large. 

Furthermore, it is a very natural question to ask whether there is a fundamental reason why the known techniques get stuck at the near-linear barrier.
One important aspect of our work is, from the theory perspective, that it breaks this threshold and thereby opens up a whole new unexplored domain of research questions.

\paragraph{Low-Memory $\mpc$ Model}
We study a more realistic regime of the parameters for problems on large graphs, captured by the following \emph{low-memory $\mpc$ model}. 


\vspace{2pt}
\begin{mdframed}[backgroundcolor=gray!20,topline=false,
  rightline=false,
  leftline=false,bottomline=false] 
\textbf{Low-Memory $\mpc$ Model for Graph Problems:}

\noindent The input is a graph $G = (V, E)$ with $n$ nodes and $m$ edges of size $N=\widetilde{O}(n+m)$. Given $M=\widetilde{O}\left(\frac{N^{1+\alpha'}}{S}\right)$ machines with local memory $S=\widetilde{O}\left(n^{\eps}\right)$ each, for arbitrary constants $\eps>0$ and $\alpha'\geq 0$, we raise the question of what problems on $G$ can be solved efficiently---that is, in $\poly \log \log n$ rounds.
\end{mdframed}
\vspace{2pt}   

Note that for sparse graphs, this condition exactly matches the sublinear memory constraint, and hence does not allow a trivial solution, as opposed to the setting with linear memory. We point out that low memory variants of the MPC model have been studied before \cite{pietracaprina2012space,ceccarello2015space}, resulting in $O(\log n)$-round algorithms for a variety of problems. For many of the fundamental graph problems, however, $O(\log n)$ is often particularly easy to achieve, for instance by directly adopting $\mathsf{LOCAL}$ algorithms.
We thus restrict our attention to ``efficient'' algorithms, which we define to be a $\poly \log \log n$ function, given that the state-of-the-art algorithms in the $\mpc$ model tend to end in this regime of round complexities. 
Note that no general super-constant lower bounds are known \cite{roughgarden2016shuffles}. 

\paragraph{Concurrent Related Work} Until very recently, $\mpc$ research had focused on linear-memory  algorithms.
After (a preliminary version of) this work, the low-memory setting gained a lot of attention. This led to a variety of new results for graph problems in this model. We briefly outline recent developments that have taken place after this work. In a follow-up work, \cite{bfukarp} devise $\mis$ and matching algorithms in uniformly sparse graphs in $O(\log^2 \log n)$ rounds. 
In independent concurrent works, Ghaffari and Uitto \cite{GU18} and Onak \cite{DBLP:journals/corr/abs-1807-08745} provide algorithms for the problems of maximal independent set and matching in general graphs in $\widetilde{O}(\sqrt{\log n})$ rounds.
In \cite{lowmemcol}, Chang et al. develop an $O(\sqrt{\log \log n})$-round low-memory $\mpc$ algorithm for $(\Delta+1)$-list coloring. 



\begin{remark}\label[remark]{smallerThanDelta}
If a node cannot be stored on a single machine, as its degree is larger than $S$, one has to introduce some sort of a workaround, e.g., have several smaller-degree copies of the same node on several separate machines.
In the end of \cref{sec:overview}, we argue how to get rid of this issue, in our problem setting, by a clean-up phase in the very beginning. To make the statements and arguments more readable, 
we throughout think of this clean-up as having taken place already. Instead, one could also work with the simplifying assumption that every machine has $S=\widetilde{O}\left(n^{\eps}+\Delta\right)$ memory, so that this issue does not arise in the first place. 
\end{remark}


\subsection{Limitations of Linear-Memory $\mpc$ Techniques}\label[section]{sec:MPCTechniques}
In the following, we briefly overview recent techniques from the world of Massive Parallel Computation algorithms, and give some indications as to why they are likely to fail in the low-memory setting. 
The restriction to substantially sublinear memory, to the best of our knowledge, indeed rules out all the known $\mpc$ techniques, which seem to hit a boundary at roughly $S=\widetilde{\Omega}(n)$: moving from essentially linear to significantly sublinear memory incurs an exponential overhead in their round complexity, 
 regardless of the density of the graph. 
This blow-up in the running time gives rise to the question of to what extent this near-linear memory is necessary for efficient algorithms.



\paragraph{(Direct) PRAM/LOCAL Simulation}
One easy way of devising $\mpc$ algorithms is by shoehorning parallel or distributed 
algorithms into the $\mpc$ setting. 
For not too resource-heavy PRAM algorithms, there is also a standard simulation technique \cite{karloff2010model,goodrich2011sorting} that automatically transforms them into $\mpc$ algorithms.
This approach, however, suffers from several shortcomings. First and foremost, the reduction leads to an $\Omega(\log n)$ round complexity, which is exponentially above our efficiency threshold.

\paragraph{Round Compression}
Another similar technique, called \emph{round compression}, introduced by Assadi and Khanna \cite{AssadiK17,Assadi2017}, provides a generic way of compressing several rounds of a distributed algorithm into fewer $\mpc$ rounds, resulting in an (almost) exponential speed-up. However, this method heavily relies on storing intermediate values, leading to a blow-up of the memory. In particular, when requiring the algorithm to run in $\poly \log \log n$ rounds, superlinear memory per machine seems inevitable. 



\paragraph{Filtering}
The idea of the \emph{filtering} technique~\cite{lattanzi2011filtering,kumar2015fast} is to reduce the size of the input by carefully removing a large amount of edges from the input graph that do not contribute to the (optimal) solution of the problem.
This reduction is done by either randomly sampling the edges, or by deterministically choosing sets of relevant edges,
so that the resulting (partial) problem instance fits on a single machine, and hence can be solved there. This requires significantly superlinear memory, or logarithmically many rounds if memory is getting close to the linear regime. Moreover, the approach seems to get stuck fundamentally at $S=\widetilde{\Omega}(n)$, since it relies on one machine eventually seeing the whole (filtered) graph. 

%

\paragraph{Coresets}
One very recent and promising direction for $\mpc$ graph algorithms is the one of (randomized composable) coresets \cite{AssadiK17,assadi2017coresets}, in some sense building on the filtering approach. The idea is that not all the information of the graph is needed to (approximately) solve the problem. One thus can get rid of unimportant parts of the information. 
Solving the problem on this core, one then can derive a perfect solution or a good approximation to it, at much lower cost. This solution, however, is found by loading the coreset (or parts of it) on one machine, and then locally computing a solution, which again seems to be stuck at $S=\widetilde{\Omega}(n)$, for similar reasons as the filtering approach.   



\subsection{Local Techniques for Low-Memory $\mpc$}
In this section, we propose a direction that seems to be promising to pursue in order to devise efficient $\mpc$ algorithms in the substantially sublinear memory regime.

\paragraph{Inherent Locality  and Local Algorithms}
The low-memory $\mpc$ model, as compared to the traditional $\mpc$ graph model and the streaming setting, suffers from inherent locality: Since the memory of a single machine is too small to fit all the nodes simultaneously, it will never be able to have a global view of all the nodes in the graph. When devising techniques, we thus need to deal with this intrinsic local view of the machines. It seems natural to borrow ideas from local distributed graph algorithms, which are designed exactly to cope with this locality restriction. A direct simulation, however, in most cases only results in  $\Omega(\poly \log n)$-round algorithms. The problem is that these algorithms do not make use of the additional power of the $\mpc$ model, the global all-to-all communication, as the communication in those message-passing-based models is restricted to neighboring nodes. 

\paragraph{Local Meets Global}
We propose a strikingly simple technique to enhance local-inspired approaches with global communication, in order to arrive at efficient algorithms in the world of low-memory $\mpc$ which are exponentially faster than their local counterparts and whose memory requirements are polynomially smaller per machine than their traditional $\mpc$ counterparts. We describe this technique in the context of the $\mis$ problem on trees, even though it is more general.

\subsection{Our Results}
 

In this paper, we focus on the \emph{Maximal Independent Set ($\mis$)} problem, one of the most fundamental local graph problems. We propose efficient and surprisingly simple algorithms for the case of trees, which is particularly interesting for the following reason. While trees admit a trivial solution in the linear-memory MPC model, this cheat will not work in our low-memory setting. In some sense, it thus is the easiest non-trivial case, which makes it the most natural starting point for further studies. 
In fact, we strongly believe that our techniques can be extended to more general graph families and problems\footnote{Indeed, there is a follow-up work generalizing our approach from trees to uniformly sparse graphs and from $\mis$ only to $\mis$ and maximal matching \cite{bfu,bfukarp}.}. 


We provide two different efficient algorithms for $\mis$ on trees. Our first algorithm in \Cref{thm} is strikingly simple and intuitive, but comes with a small overhead in the total memory of the system, meaning that $M\cdot S$ is superlinear in the input size $n$. 
\begin{theorem}\label{thm}
There is an $O(\log^2 \log n)$-round $\mpc$ algorithm that w.h.p.\footnote{As usual, w.h.p.\ stands for \emph{with high probability}, and means with probability at least $1-n^{-c}$, for any $c\geq 1$.}\ computes an $\mis$ on $n$-node trees in the low-memory setting, that is, with $S=\widetilde{O}\left(n^{\eps}\right)$ local memory on each of $M=\widetilde{O}\left(n^{1-\eps/3}\right)$ machines, for any $0< \eps <1$. 
\end{theorem}

Our second algorithm in \Cref{thm2} gets rid of this overhead at the cost of a factor of $\log \log n$ in the running time.

\begin{theorem}\label{thm2}
There is an $O(\log^3 \log n)$-round $\mpc$ algorithm that w.h.p.\ computes an $\mis$ on $n$-node trees in the low-memory setting, that is, with $S=\widetilde{O}\left(n^{\eps}\right)$ local memory on each of $M=\widetilde{O}\left(n^{1-\eps}\right)$ machines, for any $0< \eps <1$. 
\end{theorem}

The algorithms in \Cref{thm,thm2} almost match the conditional lower bound of $\Omega(\log \log n)$ for $\mis$ (on general graphs) due to Ghaffari, Kuhn, and Uitto \cite{gku19}, which holds unless there is an $o(\log n)$-round low-memory $\mpc$ algorithm for connected components. This, in turn, is believed to be impossible under a popular conjecture \cite{DBLP:conf/icml/YaroslavtsevV18}. 

Our algorithms improve almost exponentially on the $\widetilde{O}(\sqrt{\log n})$-round low-memory $\mpc$ algorithms in concurrent works---for bounded-arboricity by Onak \cite{DBLP:journals/corr/abs-1807-08745} and for general graphs by Ghaffari and Uitto \cite{GU18}---as well as on the algorithms directly adopted from the $\mathsf{PRAM}$/$\mathsf{LOCAL}$ model: an $O(\log n)$-round algorithm for general graphs due to Luby \cite{luby1986simple} and independently Alon, Babai, Itai \cite{alon1986fast}, and the $O(\sqrt{\log n} \cdot \log \log n)$-round algorithm for trees by Lenzen and Wattenhofer \cite{lenzen2011mis}. Note that for rooted trees, the $\mathsf{PRAM}$/$\mathsf{LOCAL}$ algorithm by Cole and Vishkin \cite{ColeVishkin} directly gives rise to an $O(\log^* n)$-round low-memory $\mpc$ algorithm.

Moreover, our result shows that the local  memory can be reduced substantially from $\widetilde{\Omega}(n)$ to $n^{\eps}$ or even $n^{1/\poly \log \log n}$ (see \Cref{cor}) while not incurring a significant loss in the round complexity, compared to the recent $O(\log \log n)$-round $\mis$ algorithm by Ghaffari et al.~\cite{MPCMIS}. 


Throughout the paper, when we mention the low-memory $\mpc$ setting, we refer to the parameter range for $S$ as given in \Cref{thm,thm2}, that is, $S=\widetilde{O}(n^{\alpha})$, where $\alpha>0$ is an arbitrary constant. 
However, $\eps$ does not need to be a constant. Indeed, we can even go to subpolynomial memory $S=n^{o(1)}$.

\begin{corollary}\label[corollary]{cor} For any $\eps=\Omega\left(1/\poly\log\log n\right)$, an $\mis$ on an $n$-node tree can be computed on $M=\widetilde{O}\left(n^{1-\eps/3}\right)$ machines with $S=\widetilde{O}\left(n^{\eps}\right)$ local memory each in $O\left(\frac{1}{\eps} \cdot \log^2 \log n\right)$ MPC rounds. 
\end{corollary}

\subsection{Our Approach in a Nutshell}
 In the following, we give a short (and slightly imprecise) sketch of the steps of our algorithm. Our approach is based on the \emph{shattering} technique which recently has gained a lot of attention in the $\local$ model of distributed computing \cite{barenboim2016locality} and goes back to the early nineties \cite{beck1991LLL}. The idea of \emph{shattering} is to randomly break the graph into several significantly smaller components by computing a partial solution. The problem on the remaining components then is solved by a \emph{post-shattering} algorithm. 


\subsection*{Shattering}
The goal of our \emph{shattering} technique is to compute an independent set such that after the removal of these independent set nodes and all their neighbors, the remaining graph, w.h.p., consists of components of size at most $\poly \log n$. This is done in two steps: first, the maximum degree, w.h.p., is reduced to $\poly \log n$ using the \emph{iterated subsample-and-conquer} method, and then a local shattering algorithm is applied to this low-degree graph.

\para{I) Degree Reduction via Iterated Subsample-and-Conquer}

\noindent Our subsample-and-conquer method will w.h.p.~reduce the maximum degree of a graph polynomially, from $\Delta$ to roughly $\Delta^{1/(1+\eps)}$, as long as $\Delta =\Omega(\poly \log n)$.
After $O(\log_{1+\eps} \log \Delta)$ iterations, the degree of our graph drops to $\poly \log n$.



\paragraph{Subsample}
We sample the nodes independently with probability roughly $\Delta^{-\frac{1}{1+\eps}}$, where $\Delta$ is an upper bound on the current maximum degree\footnote{Note that in the $\mpc$ model it is easy to keep track of the maximum degree.} of the graph.
This subsampling step guarantees, roughly speaking, the following three very desirable properties of the graph $G'$ induced by the sampled nodes.
\begin{enumerate}[i)]
	\item The diameter of each connected component of $G'$ is bounded by $\bigO(\log_\Delta n)$.
	\item The number of nodes in each connected component of $G'$ is at most $n^{\eps/3}$.
	\item Every node with degree $\Delta^{1/(1+\eps)}$  or higher in $G$ has many neighbors in $G'$.
\end{enumerate}
\paragraph{Conquer}
We find a random $\mis$ in all the connected components of $G'$ in parallel. This can be done by gathering the connected components\footnote{Gathering the connected components means loading all the nodes of a connected component onto the same machine.}, locally picking one of the two 2-colorings of this tree uniformly at random, and adding the black, say, nodes to the $\mis$. We will see that properties i) and ii) are crucial to ensure that the gathering can be done efficiently. In particular, storing the components on a single machine is possible due to the small size of the components, and the gathering is fast due to the small diameter.
Because of property iii), the randomness in the choice of the $\mis$ in every connected component, as well as the tree structure, all high-degree nodes in the original graph (sampled or not), w.h.p., will have an adjacent independent set node and thus, are removed from the graph for the next iteration.

\para{II) Low-Degree Local Shattering}

\noindent Once the degree has dropped to $\Delta'=\poly \log n$, we apply the shattering part of the $\local$ $\mis$ algorithm of Ghaffari~\cite{Ghaffari-MIS}, which runs in $O(\log \Delta')=O(\log \log n)$ rounds and w.h.p.~leads to connected components of size $\poly \Delta' \cdot \log n=\poly\log n$ in the remainder graph. Observe that the simulation of this algorithm in the $\mpc$ model is straightforward.

\subsection*{Post-Shattering}
We gather the connected components of size $\poly \log n$ and solve the remaining problem locally. 



\section{Algorithm Overview and Roadmap}\label[section]{sec:overview}
In this section, we give the formal statements we need to prove our main result, and provide an overview of the structure of the remainder of the paper.
We start with a result that is repeatedly used to gather all nodes of a connected component onto one machine, provided that they fit there. It will come in two variants, which naturally give rise to \cref{thm,thm2}, respectively. The proof is deferred to \Cref{sec:CC} (part a)) and the full version \cite{full} (part b)).
\begin{lemma}[Gathering]\label[lemma]{lemma:hashToMin}
Let $G$ be an $n$-node graph and $G'$ any $n'$-node subgraph of $G$ consisting of connected components of size at most $k=O\left(n^{\eps/3}\right)$ and diameter at most $d$. Then there are
\begin{enumerate}[a)]
\item an $O(\log d)$-round low-memory $\mpc$ algorithm with $M=\widetilde{O}\left(n^{1-\eps/3}\right)$ machines and 
\item an $O(\log d \cdot \log \log n)$-round low-memory $\mpc$ algorithm with $M=\widetilde{O}\left(n^{1 - \eps}\right)$ machines, if $n' \cdot d^3 = O( n )$, 
\end{enumerate}
that compute an assignment of nodes to machines so that all the nodes of a connected component of $G'$ are on the same machine.
\end{lemma}

Next, we will provide the results corresponding to the two main parts of our algorithm, the \emph{shattering} and the \emph{post-shattering}.

\begin{lemma}[Shattering]\label[lemma]{lemma:shattering}
There are  
\begin{enumerate}[a)]
\item an $O(\log \log n \cdot \log \log \Delta)$-round low-memory $\mpc$ algorithm that uses $M=\widetilde{O}(n^{1-\alpha/3})$ machines and
\item an $O(\log^2 \log n \cdot \log \log \Delta)$-round low-memory $\mpc$ algorithm with $M=\widetilde{O}(n^{1-\alpha})$ machines
\end{enumerate}
that compute an independent set on an $n$-node tree with maximum degree $\Delta$ so that the remainder graph, after removal of the independent set nodes and their neighbors, w.h.p.~has only components of size at most $\poly \log n$. 
\end{lemma}

The proof of this \emph{Shattering Lemma} can be found in \Cref{sec:shattering}. The following \emph{Post-Shattering Lemma} is a direct consequence of the \emph{Gathering Lemma}.


\begin{lemma}[Post-Shattering]\label[lemma]{lemma:postShattering}
There are
\begin{enumerate}[a)]
\item an $O(\log k)$-round low-memory $\mpc$ algorithm with $M=\widetilde{O}\left( n^{1 - \eps/3} \right)$ machines and
\item an $O(\log k \cdot \log \log n)$-round low-memory $\mpc$ algorithm with $M=\widetilde{O}\left(n^{1 - \eps} \right)$ machines
\end{enumerate}
that find an $\mis$ in an $n$-node graph consisting of connected components of size $k=O\left(n^{\eps/3}\right)$. 
\end{lemma}
\begin{proof}
By \Cref{lemma:hashToMin}, we can gather the connected components in $O(\log k)$ rounds. 
Then, an $\mis$ of each connected component can be computed locally.
Note that Theorem 1.1 by Ghaffari~\cite{Ghaffari-MIS} certifies that the number of nodes remaining after our shattering process can be made small enough to satisfy the conditions required by \Cref{lemma:hashToMin}.
%
\end{proof}
Note that the naive simulation of the corresponding $\local$ post-shattering algorithm~\cite{Ghaffari-MIS,panconesi1992improved} would lead to a round complexity of $2^{\bigO(\sqrt{\log \log n})}$.





We now put together the results to prove \Cref{thm,thm2}. 

\begin{proof}[Proof of \Cref{thm,thm2}]
We apply the \emph{shattering} algorithm from \Cref{lemma:shattering} to get an independent set, with connected components of size $k=\poly \log n$ in the remainder graph. Then we run the \emph{post-shattering} algorithm from \Cref{lemma:postShattering} to find an $\mis$ in all these components. The combination of the initial independent set found by the \emph{shattering} and all the $\mis$ found by the \emph{post-shattering} results in an $\mis$ in the original tree. 
\end{proof}


\paragraph{Memory per Machine below $\Delta$}
If the degree of a node is larger than the local memory, one needs to store several lower-degree copies of this node on different machines. Here, we give a short argument for why one can assume without loss of generality that all incident edges of a node are stored on the same machine. 
Notice that in a tree with $n$ nodes, there can be at most $n^{1 - \eps/2}$ nodes with degree at least $n^{\eps/2}$.
If we now just ignore all these high-degree nodes and find an $\mis$ among the remaining nodes, the resulting graph, after removal of all $\mis$ nodes and their neighbors, has at most $n^{1 - \eps/2}$ nodes.
Repeating this argument roughly $2/\eps$ times gives an $\mis$ in the whole input graph.
%

\section{Shattering}\label[section]{sec:shattering}

\begin{lemma}[Iterated Subsample-and-Conquer]\label[lemma]{lemma:degreeAndDiam}
There are 
\begin{enumerate}[a)]
\item an $O(\log_{1+\eps} \log \Delta)$-round low-memory $\mpc$ algorithm with $M=\widetilde{O}(n^{1-\alpha/3})$ machines and
\item an $O(\log_{1+\eps} \log \Delta\cdot \log \log n)$-round low-memory $\mpc$ algorithm with $M=\widetilde{O}(n^{1-\alpha})$ machines 
\end{enumerate}
that compute an independent set on an $n$-node tree with maximum degree $\Delta$ such that the remainder graph, after removal of the independent set nodes and their neighbors, w.h.p.~has maximum degree $\poly \log n$. 
\end{lemma}

The proof of this lemma can be found in \Cref{sec:sampling}.

\begin{lemma}[Low-Degree Local Shattering \cite{Ghaffari-MIS}]\label[lemma]{lemma:ghaffariMIS}

There is an $O(\log \Delta)$-round LOCAL algorithm that computes an independent set on an $n$-node graph with maximum degree $\Delta$ so that the remainder graph, after removal of all nodes in the independent set and their neighbors, w.h.p.~has connected components of size $\poly \Delta \cdot \log n$.
\end{lemma}

We now combine these two results to prove \Cref{lemma:shattering}.

\begin{proof}[Proof of \Cref{lemma:shattering}]
We apply the algorithm of \Cref{lemma:degreeAndDiam}, w.h.p.~yielding an independent set with a remainder graph that has maximum degree $\Delta'=\poly \log n$. On this low-degree graph, we simulate the LOCAL algorithm of \Cref{lemma:ghaffariMIS} in a straight-forward manner, which takes $O(\log \Delta')=O(\log \log n)$ rounds and w.h.p.\ leaves us with connected components of size $\poly \Delta' \cdot \log n = \poly \log n$. 
\end{proof}

\subsection{Degree Reduction via Iterated Subsampling}\label[section]{sec:sampling}
We prove the following result, and then show how it can be used to prove \Cref{lemma:degreeAndDiam}. 
For the purposes of the proof of \Cref{lemma:degdrop} we assume that $\Delta$ is a large enough $\poly \log n$ in order to be able to apply \Cref{lemma:hashToMin}.
Notice that from the perspective of the final runtime, the exponent of the logarithm turns into a constant factor hidden in the $O$-notation.

\begin{lemma}
There are 
\begin{enumerate}[a)]
\item
an $O(\log \log n)$-round low-memory $\mpc$ algorithm with $M=\widetilde{O}(n^{1-\alpha/3})$ machines and
\item an $O(\log^2 \log n)$-round low-memory $\mpc$ algorithm with $M=\widetilde{O}(n^{1-\alpha})$ machines
\end{enumerate}
that compute an independent set on an $n$-node tree $G$ with maximum degree $\Delta=\Omega(\poly \log n)$ such that the remainder graph, after removal of the independent set nodes and their neighbors, w.h.p.~has maximum degree at most $\Delta^{(1+\delta')\delta}$, for some $\delta=\Theta\left(1/(1+\eps)\right)$ and any $\delta'>0$.
\label[lemma]{lemma:degdrop}
\end{lemma}
\begin{proof}
We first outline the algorithm and then slowly go through the steps of the algorithm again while proving its key properties. 
 
\paragraph{Algorithm} Every node is sampled independently with probability $\Delta^{-\delta}$ into a set $V'$. The connected components of $G'=G[V']$ are gathered by \Cref{lemma:hashToMin}, and one of the two 2-colorings is picked uniformly at random, independently for every connected component. This can be done locally. All the black nodes, say, are added to the $\mis$, and are removed from the graph along with their neighbors. 
See \cref{fig:sampling} in \cref{inspacegather}.

\begin{figure}
	\centering
	\includegraphics[scale=1.5]{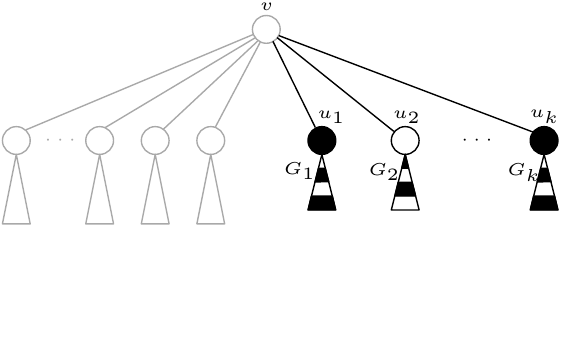}
	\caption{Consider a non-subsampled high-degree node $v$. While most of its neighbors will not be subsampled (gray nodes), there will still be many, say $k$, subsampled neighbors. 
	Since we are in a tree, all the subtrees $G_1, \ldots G_k$ of subsampled neighbors of $v$ are disjoint and thus colored independently. Node $v$ will thus, w.h.p., see at least one black neighbor. This node will be added to the $\mis$, leading to the removal of $v$ from the graph.}
	\label{fig:sampling}
\end{figure}

\paragraph{Subsampling}
We first prove that the random subsampling leads to nice properties of the graph induced by subsampled nodes. 
\begin{claim}\label[claim]{claim:smallSampled}
After the subsampling, w.h.p., the following holds.
\begin{enumerate}[i)]
\item Every connected component of $G'$ has diameter $O\left( \frac{1}{\delta} \cdot \log_{\Delta} n\right)$. 
\item Every connected component of $G'$ consists of $n^{O\left((1-\delta)/\delta\right)}$ nodes.
\item Every node with degree $\Omega\left(\Delta^{(1+\delta')\delta}\right)$ in $G$ has degree $\Omega(\poly \log n)$ in $G'$. 
\end{enumerate}
\end{claim}
\begin{proof}Consider an arbitrary path of length $\ell=\Omega\left(\frac{1}{\delta}\cdot \log_{\Delta} n\right)$ in $G$. This path is in $G'$ only if all its nodes are subsampled into $V'$, which happens with probability at most $\Delta^{- \delta\cdot\ell}=\frac{1}{\poly n}$. A union bound over all---at most $n^2$ many---paths in the tree $T$ shows that, w.h.p., the length of every path, and hence in particular also the diameter of every connected component, in $G'$  is bounded by $O\left(\frac{1}{\delta}\cdot \log_{\Delta} n\right)$. 
Since the degree among the subsampled nodes is bounded by $O\left(\Delta^{1-\delta}\right)$, w.h.p., which is a simple application of Chernoff and union bound, 
it follows that every connected component consists of at most $O\left(\Delta^{(1-\delta) \cdot \ell}\right)=n^{O\left((1-\delta)/\delta\right)}$ nodes.  
Finally, another simple Chernoff and union bound argument shows that every node with degree $\Omega\left(\Delta^{(1+\delta')\delta}\right)$  in the graph $G$ has at least $\Omega\left(\Delta^{\delta'\cdot \delta}\right)=\Omega(\poly \log n)$ neighbors in $G'$, which concludes the proof of \Cref{claim:smallSampled}.
\end{proof}
%
%
\paragraph{Gathering}
Since $G'$ consists of components that have a low diameter by \Cref{claim:smallSampled} i) and that are small enough to fit on a single machine by \Cref{claim:smallSampled} ii)---provided that $\delta=\Theta\left(1/(1+\eps)\right)$ is chosen such that the components have size $O\left(n^{\eps/3}\right)$---we can gather them efficiently by \Cref{lemma:hashToMin}, in either $O(\log \log n)$ or $O(\log ^2 \log n)$ rounds. The random $\mis$ can then be easily computed locally. 


\paragraph{Random $\mis$}
It remains to show that every high-degree node in $G$, w.h.p., has at least one adjacent node that joins the random $\mis$, which leads to the removal of this high-degree node from the graph. Note that this is trivially true for all subsampled nodes, by maximality of an $\mis$.

Now consider an arbitrary non-subsampled node $v$ with degree $\Omega\left(\Delta^{(1+\delta')\delta}\right)$ and its $\Omega(\poly\log n)$ subsampled neighbors, by \Cref{claim:smallSampled} iii). Observe that, since we are in a tree and thus in particular in a triangle-free graph, there cannot be edges between these neighbors. Therefore no two neighbors of a non-subsampled node belong to the same connected component in $G'$, which means that all the neighbors in $V'$ of $v$ are colored independently, and hence are added to an $\mis$ independently with probability $1/2$. By the Chernoff inequality, w.h.p.~at least one of $v$'s neighbors must have been added to an $\mis$, and a union bound over all nodes concludes the proof of the degree reduction, and hence of \Cref{lemma:degdrop}. 
\end{proof}
\begin{proof}[Proof of \Cref{lemma:degreeAndDiam}]
Follows from $\log_{\frac{1}{(1+\delta')\delta}}\log \Delta=\log_{1+\eps}\log \Delta$ many applications of \Cref{lemma:degdrop}.
\end{proof}

\section{Gathering Connected Components}\label[section]{sec:CC}

In this section, we provide a proof of the \emph{Gathering Lemma} in \cref{lemma:hashToMin}. Our approach is essentially a tuned version of the Hash-to-Min algorithm by Chitnis et al.~\cite{Chitnis2013CCMapReduce} and the graph exponentiation idea by Lenzen and Wattenhofer~\cite{Lenzen2010brief}. Notice that, however, Chitnis et al.~only show an $\bigO(\log n)$ bound for the round complexity; it is not possible to just use their method as a black box.
The section is divided into two subsections, where we first give a simple and fast but memory-inefficient algorithm and then present a slightly slower algorithm that only needs a constant space overhead.

In very recent works, independent of this paper, Andoni et al.~\cite{Andoni2018} and Assadi et al.~\cite{Assadi2018} studied, among other problems, finding connected components in the low-memory setting of MPC.
In particular, Andoni et al.\ give algorithms to find connected components and to root a forest with constant success probability, with $O(m)$ total memory in time $O(\log d \cdot \log \log n)$.
While their results are more general, ours have the advantages of being (arguably) much simpler and deterministic.
Furthermore, to turn their algorithm to work with high probability, the straightforward approach requires a logarithmic overhead in the total memory.

We present the naive gathering algorithm in \cref{sec1} and the in-space gathering in \cref{inspacegather}.

\subsection{Naive Gathering}\label{sec1}
We first present the algorithm. The underlying idea of the algorithm is to find a minimum-ID\footnote{We assume without loss of generality that every node has a unique identifier. If not, every node can draw an $O(\log n)$-bit identifier at random, which w.h.p.\ will be unique.} node within every component and to create a virtual graph that connects all the nodes of that component to this minimum-ID node, the \emph{leader}.

\paragraph{Gathering Algorithm}
In every round, every node $u$ completes its $1$-hop neighborhood to a clique. Once a round is reached in which there are no more edges to be added, $u$ stops and selects its minimum-ID neighbor as its leader.
We refer to \cref{fig:hashmin} in \cref{inspacegather} for an illustration.

\begin{figure*}
	\centering
	\begin{subfigure}{.45\textwidth}

		\centering

		\includegraphics[scale=1]{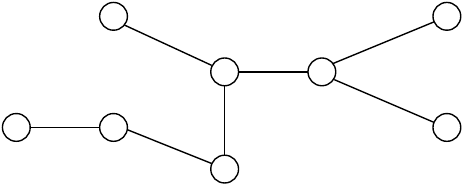}

		\caption{}\label{fig:exp1}

	\end{subfigure}
	\begin{subfigure}{.45\textwidth}

		\centering

		\includegraphics[scale=1]{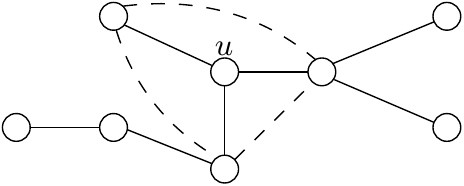}

		\caption{}\label{fig:exp2}

	\end{subfigure}
	\begin{subfigure}{.45\textwidth}

		\centering

		\includegraphics[scale=1]{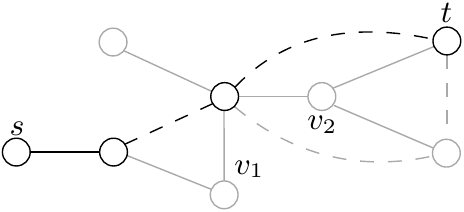}

		\caption{}\label{fig:exp3}

	\end{subfigure}
	\caption{We illustrate the gathering algorithm with help of the tree depicted in \cref{fig:exp1}. The edges added by node $u$ are illustrated in \cref{fig:exp2} by dashed arcs. \cref{fig:exp3} displays how the edges added by nodes $v_1$ and $v_2$, drawn as dashed arcs, shortcut the shortest path between nodes $s$ and $t$.}
	\label{fig:hashmin}
\end{figure*}

Observe that once there is a round in which $u$ does not add any edges, the component of $u$ forms a clique, and thus all nodes in this component have the same leader, namely the minimum-ID node in this clique. 
Next, we prove that this algorithm terminates quickly. 
\begin{claim}
	The gathering algorithm takes $\bigO(\log d)$ rounds on a graph with diameter $d$.
	\label[claim]{lemma:htaruntime}
\end{claim}
\begin{proof}
	Consider any shortest path $u_1, \ldots, u_\ell$ of length $2 \leq \ell \leq d$.
	After the first round, every $u_i$ gets connected to $u_{i - 2}$ and $u_{i + 2}$ for $2 < i < \ell - 1$.
	Thus, the diameter of the new graph is at most $\lceil 2d/3 \rceil$.
	After $\bigO(\log d)$ iterations, the diameter within each component has reduced to 1,  and the algorithm halts.	
\end{proof}

It remains to show that not too many edges are added, so that the virtual graph of any component still fits into the memory of a single machine. 
\begin{claim}The number of edges in the virtual graph created by the gathering algorithm in a component of size $k$ is $\bigO(k^{3})$. 
	\label[claim]{lemma:componentsize}
\end{claim}
\begin{proof}
	During the execution of the algorithm, each node in a component may create an edge between any other two nodes in the corresponding component, thus at most $k^3$.
\end{proof}

Since we require the components to be of size at most $O(n^{\eps/3})$, the previous claim guarantees that the virtual graph of any connected component indeed fits into the memory. 
So as to not overload any machine with too many components, we assume that the shuffling distributes the components to the machines in an arbitrary feasible way, e.g., greedily\footnote{An alternative and simple way to prevent overloading is to add an $\bigO(\log n)$ factor of memory per machine and consider a random assignment of components to machines as a balls-into-bins process.}.

\begin{remark}
A weakness of the gathering algorithm is that we need $O(k^3)$ memory to store a connected component of size $k$, even if this component originally just consisted of as few as $k-1$ edges. This is because a single edge can exist on up to $k$ machines. In the worst case, the required memory is blown up by a power 3.  
This leads to a super-linear overall memory requirement, that is, we need roughly $N^{1 + 2\eps/3}$ total memory in the system.
Notice that this can be implemented either by adding more machines or by adding more memory to the machines, since we do not care on which machines the resulting components lie, as long as they fit the memory.
\end{remark}

\subsection{In-Space-Gathering in Trees}\label{inspacegather}
The simple and naive gathering algorithm can be very wasteful in terms of space usage over the whole system.
In this section, we provide a fine-tuned version of the gathering method that works, asymptotically, in space, thereby proving part b) of \Cref{lemma:hashToMin}.
In other words, the total space requirement drops to $\bigO(n)$.
Informally, our algorithm first turns every connected component into a rooted tree and then determines which nodes are contained in the same tree component by making sure that each node learns the ID of the root of its tree.
For the latter part, we prove the following.

%

\begin{lemma}
	\label[lemma]{lemma:rootedForest}
	There is an $O(\log d)$-round low-memory MPC algorithm that works in an $n$-node forest of rooted trees with maximum diameter $d$ and, for every node, determines the root of the corresponding tree. The algorithm requires $M = O\left(n^{1 - \eps}\right)$ machines.
\end{lemma}
\begin{proof}
	Let $\parent(v)$ denote the parent of node $v$ and define $\parent(r) = r$ for a root node $r$.
	Consider the following pointer-forwarding algorithm that is run in parallel for every node $v$.	
	In every round, for every child $u$ of $v$, we set $\parent(u) := \parent(v)$.
	The process terminates once $v$ points to a root, i.e., to a node $r$ for which $\parent(r) = r$.
	Notice that after every step, following the parent pointers still leads to the root node.
	
	Let $(v_1, v_2, \ldots, v_k)$ be the directed path from node $v_1$ to the root $r = v_k$ of its subtree in round $t$.
	After one round of the algorithm, every $v_i$ is connected to $v_{\min\{ i + 2, k \}}$.
	Thus, the length of the path is at most $\lceil k/2 \rceil$.
	After $\bigO(\log k) = \bigO(\log d)$ rounds the algorithm terminates yielding the claim. 	
\end{proof}

\paragraph{Root a Tree}
Given \Cref{lemma:rootedForest}, what remains to show for our algorithm is how to root a tree.
The idea is to once more use the graph exponentiation method to learn an $\ell$-hop neighborhood of a node in $\log \ell$ steps.
However, in order to prevent the space requirement from getting out of hand, each node performs only a bounded number of exponentiation steps, after which all nodes that already know their parent in the output orientation are removed from the graph.
Then this process is iterated until at most one node (per connected component) remains.

%
%

\paragraph{Tree-Rooting Algorithm $\mathcal A$}
In the following, we give a formal description of an algorithm $\mathcal A$ for rooting a tree of diameter $d$.
The algorithm takes an integer $B$ as input parameter that describes the initial memory \emph{budget} for each node $v$, i.e., an upper bound on the number of edges that $v$ may add before the first node removal.
The execution of $\mathcal A$ is subdivided in phases $i = 0,1,\ldots$ which consist of $O(\log d)$ rounds each.
Set $B_0 = B$.

\paragraph{Phase $i$ of $\mathcal A$}
In phase $i$, each node $v$ does the following:

In round $0$, node $v$ sets its local budget $B_v$ to $B_i$.
In each following round $j = 1, 2, \ldots$, node $v$ first connects its $1$-hop neighborhood to a clique by adding edges between all its neighbors that are not connected yet, but it does so only if the number of added edges is at most $B_v$.
Then $v$ updates its local budget by decreasing $B_v$ by the number of edges that $v$ added.
If $B_v$ was not large enough to connect $v$'s $1$-hop neighborhood to a clique, then $v$ does not add any edges in round $j$.
This concludes the description of round $j$, of which there are $O(\log d)$ many.

Denote the tree at the beginning of phase $i$ by $T_i$, and for each neighbor $u$ of $v$, denote the set of nodes that are closer to $u$ than $v$ in $T_i$ by $S^i_u(v)$.   
Phase $i$ concludes with a number of special rounds:
First, $v$ checks whether it has a neighbor $u'$ in $T_i$ with the following properties:
\begin{enumerate}[i)]
	\item $S^i_{u}(v)$ is contained in the current $1$-hop neighborhood of $v$, for each neighbor $u$ of $v$ in $T_i$ satisfying $u \neq u'$.
	\item $S^i_{u'}(v)$ is not (entirely) contained in the current $1$-hop neighborhood of $v$.
\end{enumerate}
If such a neighbor $u'$ exists (which, by definition, is unique), then $v$ sets $\parent(v) = u'$.
Second, $v$ removes all edges that it added during phase $i$ (regardless of whether a parent is set).
Third, $v$ is removed from $T_i$ if it already chose its parent, i.e., if it set $\parent(v)$.
Fourth, the budget per node is updated, by setting $B_{i+1} = B_i \cdot n_i/n_{i+1}$, where $n_i$ and $n_{i+1}$ are the numbers of nodes of $T_i$ and $T_{i+1}$, respectively.
This concludes the description of phase $i$.

We execute this process until at most one node remains.

\paragraph{Termination of $\mathcal{A}$} Since in each phase (at the very least) all leaves are removed, this process eventually terminates.

It is straightforward to check that if a node $v$ chooses its parent $u' = \parent(v)$ in phase $i$, then any neighbor $u \neq u'$ of $v$ in $T_i$ also chooses its parent in phase $i$, and, what is more, $u$ chooses $v$ as its parent (which, combined with the following observations, shows that the orientation of the input tree induced by the parent choices of the nodes yields indeed a rooted tree).
Hence, given the above process, one of two things happens in the end: either exactly one node remains, or all nodes are removed but there is exactly one pair of nodes that chose each other as their parent.
In the former case, no action has to be taken, as the remaining node is simple the root of our rooted tree.
In order to handle the latter case, we add a simple fifth special round at the end of each phase $i$:
Each node $v$ removed in phase $i$ checks whether the node it chose as its parent chose $v$ as its parent. 
If this is the case, then the node with the higher id removes its choice of parent and becomes the root node of the input tree.
See \cref{fig:cutAlgo} for an illustration of algorithm $\mathcal A$.

\begin{figure}
	\begin{subfigure}{0.45\textwidth}
		\centering
		\includegraphics[scale=0.8]{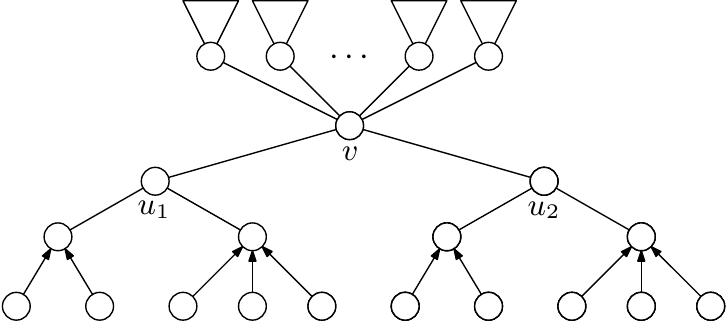}
		\caption{The leaves have only one neighbor, which becomes their parent. Nodes $u_1, u_2$, and $v$ do not have enough budget to add edges.}
		\label{subfig:cut1}
	\end{subfigure}\qquad\quad
	\begin{subfigure}{0.45\textwidth}
		\centering
		\includegraphics[scale=0.8]{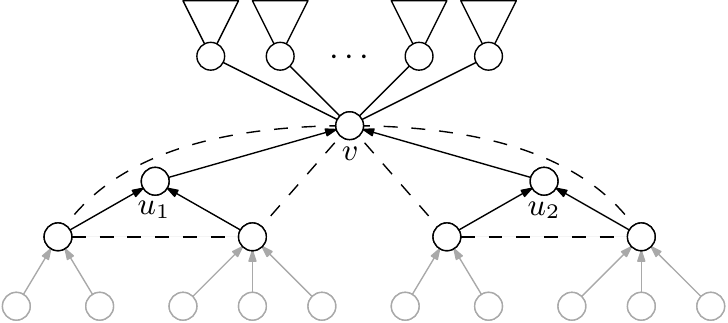}
		\caption{Once the leaves are removed once, enough budget is freed for nodes $u_1$ and $u_2$ to add edges that connect their neighbors.}
		\label{subfig:cut2}
	\end{subfigure}\\\vspace{2px}
	\begin{subfigure}{0.95\textwidth}
		\centering
		\vspace{20px}
		\includegraphics[scale=0.8]{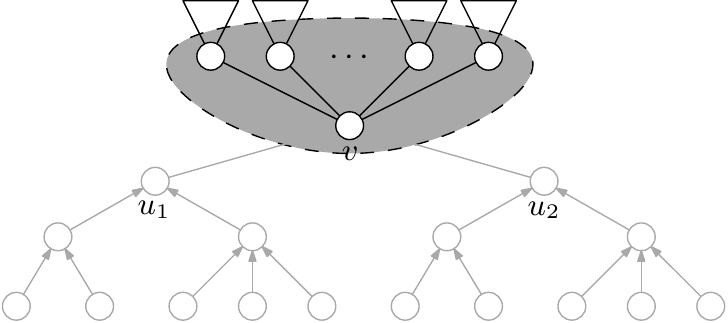}
		\caption{Once nodes $u_1$ and $u_2$ know that $v$ is their parent, node $v$ can focus its budget to the remainder of the tree illustrated by the gray area.}
		\label{subfig:cut3}
	\end{subfigure}
	
	\caption{An illustration of the (local) steps performed by algorithm $\mathcal A$. For the sake of this example, assume that the initial budget is $2$. As illustrated in \cref{subfig:cut1}, leaves are always able to determine their parent. Assuming that the tree has non-leaf minimum degree $3$, removing all the leaves at least roughly doubles the budget of all nodes. Thus, in the second step, illustrated in \cref{subfig:cut2}, nodes with degree at most $4$ are able to complete their $1$-hop neighborhoods into a clique. Small subtrees rooted at (or connected to) $v$ are removed quickly in our process, as illustrated in \cref{subfig:cut3}, and therefore, node $v$ requires large subtrees to survive for many phases.}
	\label{fig:cutAlgo}
\end{figure}
\newpage
\paragraph{Running Time of $\mathcal{A}$}
We present a number of lemmas in order to determine the runtime of algorithm $\mathcal A$. Here, a subtree $T(v)$ rooted at some node $v$ corresponds to the descendants of $v$ in the rooted tree $T$ returned by $\mathcal A$ (or in the rooted subtree of $T$ induced by the nodes of some $T_i$).
\begin{lemma}
	\label[lemma]{lemma:subtreeremoval}
	Consider some arbitrary phase $i$, and let $T(v)$ be the subtree of $T_i$ rooted at $v$. If $|T(v)| \leq \sqrt{B_i}$, then $v$ chooses its parent in phase $i$ and is removed from the tree.
\end{lemma}
\begin{proof}
	Let $k$ be some arbitrary non-negative integer, and consider any node $u$ in $T(v)$ with distance at least $2^k$ to $v$.
	Observe that, according to $\mathcal A$, the distance between any two nodes in $T_i$ decreases by a factor of at most $2$ per round.
	Hence, after round $k$ of phase $i$, all nodes contained in the $1$-hop neighborhood of $u$ are actually also contained in $T(v)$.
	Thus, each edge that $u$ would have added if it had connected its $1$-hop neighborhood to a clique in each of the rounds $1, \ldots, k+1$, disregarding any budget constraints, is an edge between nodes from $T(v)$.
	Since $|T(v)| \leq \sqrt{B_i}$, the number of edges between nodes from $T(v)$ is bounded from above by $B_i$; it follows that $u$ had enough budget to indeed connect its $1$-hop neighborhood to a clique in each round up to and including round $k+1$.

	Now consider any node $w$ whose distance to $v$ in $T_i$ is at least $2^k$, but at most $2^{k+1} - 1$.
	Let $w_0, \ldots, w_{k}$ be nodes on the unique path between $v$ and $w$ with distance $2^0, \ldots, 2^k$ to $v$.
	Due to the observations above, it is straightforward to check that, in each round $1 \le h \le k$, node $w_{h-1}$ connects node $w_h$ to $v$, while in round $k+1$, node $w_k$ connects node $u$ to $v$.
	Since the depth of $T(v)$ is upper bounded by $\log d$, it follows that after $\log d$ rounds, all nodes from $T(v)$ are contained in $v$'s $1$-hop neighborhood.
	Hence, $v$ will choose the only neighbor that is not contained in $T(v)$ as its parent, and $v$ is removed in phase $i$.
	Since $k$ was chosen arbitrarily, the lemma statement follows.
\end{proof}

\begin{lemma}
	\label[lemma]{lemma:subtrees}
	Let $T$ be a rooted tree with $n$ nodes and diameter at most $d$.
	Let $1 \leq \alpha \leq n$, and let $C$ be the set of nodes $v$ with the property that $|T(v)| \leq \alpha$.
	Then, $|C| \geq n \cdot (\alpha/(d + \alpha))$.
\end{lemma}
\begin{proof}
	Assign one dollar to each node that is not contained in a subtree of size at most $\alpha$.
	Every such node then distributes its dollar evenly among all of its descendants in $C$.
	Note that, for each leaf node $w$ of the tree obtained from $T$ by deleting all nodes in $C$, the number of descendants of $w$ in $C$ is at least $\alpha$ since otherwise $w$ would be in $C$, by the definition of $C$.
	Hence, all nodes that are not contained in $C$ have at least $\alpha$ descendants in $C$.
	
	Consider then any node $v \in C$.
	Since the diameter of the tree is $d$, node $v$ can have at most $d$ ancestors in $T$.
	Every ancestor of $v$ distributes at most $1/\alpha$ dollars to $v$ and therefore, $v$ receives at most $d / \alpha$ dollars.
	As the amount of dollars did not change during its redistribution from nodes not contained in $C$ to nodes in $C$, we can conclude that $|C| \cdot (d / \alpha) \geq n - |C|$ which implies that $|C| \geq n \cdot (\alpha/(d + \alpha))$.
\end{proof}


\begin{lemma}
	\label[lemma]{lemma:rootingruntime}
	Assume that the input parameter $B$ for our algorithm $\mathcal A$ satisfies $B \ge d^3$. Then the runtime of $\mathcal A$ on trees with $n$ nodes and diameter $d$ is $O(\log d \cdot \log \log n)$.
\end{lemma}
\begin{proof}
	Observe that the sequence $B_0, B_1, \ldots$ of budgets at the beginning of phases $0, 1, \ldots$ is monotonically non-decreasing, by definition.
	Hence, $B_i \ge d^3$ for all phases $i$.
	Now consider some arbitrary phase $i$, and let $n_i$ denote the number of nodes of $T_i$.
	By Lemma \ref{lemma:subtreeremoval} and Lemma \ref{lemma:subtrees}, the number of nodes that are removed in phase $i$ is at least $n_i \cdot (\sqrt{B_i}/(d + \sqrt{B_i}))$.
	Thus, for the new budget $B_{i+1}$, it holds by definition that
	\[
		B_{i+1} \geq B_i \cdot \frac{1}{1-\frac{\sqrt{B_i}}{d + \sqrt{B_i}}} =  B_i \cdot \left(\frac{d + \sqrt{B_i}}{d}\right) \geq \frac{B_i^{\frac{3}{2}}}{d} \enspace .
	\]
	Since, as observed above, $d \le B_i^{1/3}$, we obtain $B_{i+1} \geq B_i^{7/6}$, which implies $B_{i+5} \geq B_i^2$.
	Recall that in each phase $i$, at least a $(\sqrt{B_i}/d)$-fraction of nodes is removed.
	Thus, after $\bigO(\log \log n)$ phases, all nodes (except possibly for one node) have been removed and the termination condition of $\mathcal A$ is satisfied.
	Since every phase takes $O(\log d)$ time, the claim follows.
\end{proof}

Now we have all the ingredients to prove the second part of \Cref{lemma:hashToMin}. It is a simple corollary of the following theorem.

\begin{theorem}
	Consider a forest $F$ of $n$ nodes where every tree is of diameter at most $d$.
	There is an MPC algorithm that finds the connected components of $F$ in time $O(\log d \cdot \log \log n)$ where $M \cdot S = \bigO(n \cdot d^3)$.
\end{theorem}

\begin{proof}
	Imagine that we run algorithm $\mathcal A$ in parallel on all trees of the input forest $F$, with input parameter $B = d^3$.
	There are only two parts of $\mathcal A$ that are of a global nature, i.e., where the actions of nodes do not depend on their immediate neighborhood: the termination condition that all nodes, possibly except for one, have been removed, and the part where the node's budgets are updated from $B_i$ to $B_{i+1}$.
	The former is easily adapted to the case of forests; each node simply terminates when itself or all its neighbors are removed.
	Regarding the updating of the budget, we adapt the tree rooting algorithm as follows: we still set the new budget $B_{i+1}$ to $B_i \cdot n_i/n_{i+1}$, but now $n_i$ and $n_{i+1}$ denote the total number of nodes (i.e., in all trees of the forest) that have not been removed yet at the beginning of phase $i$, resp.\ phase $i+1$. 
	
	In the following, we verify that Lemmas \ref{lemma:subtreeremoval}, \ref{lemma:subtrees}, and \ref{lemma:rootingruntime} also hold for forests instead of trees.
	In the case of Lemma \ref{lemma:subtreeremoval}, this is obvious as the argumentation is local and thus also applies to forests.
	Lemma \ref{lemma:subtrees} trivially also holds for forests since the lemma statement holds for all trees in the forest.
	Finally, since the argumentation of the proof of Lemma \ref{lemma:rootingruntime} does not make use of the fact that the input graph is a tree except when applying Lemmas \ref{lemma:subtreeremoval} and \ref{lemma:subtrees}, it follows that Lemma \ref{lemma:rootingruntime} also holds for forests.

	Hence, our adapted tree rooting algorithm actually transforms the forest into a rooted forest in time $O(\log d \cdot \log \log n)$.
	Now we can apply Lemma \ref{lemma:rootedForest}, and, e.g., color each component with the color of the root node, thereby marking the connected components.
	Due to the runtime given in Lemma \ref{lemma:rootedForest}, our total runtime is still $O(\log d \cdot \log \log n)$.
	
	It remains to show that the claimed memory constraints are satisfied.
	Due to the space guarantee given in Lemma \ref{lemma:rootedForest}, it is sufficient to show that the memory overhead induced by adding edges during the execution of out forest rooting algorithm does not exceed the allowed amount.
	Thus, consider the number of edges added in an arbitrary phase $i$.
	Since each node adds at most as many edges as its budget allows, i.e., at most $B_i$ edges, the total number of edges added in phase $i$ is upper bounded by $n_i \cdot B_i$.
	By the definition of $B_{j+1}$, we have $n_{j+1} \cdot B_{j+1} = n_j \cdot B_j$, for any phase $j$.
	Hence, the value of $n_i \cdot B_i$ is the same for every phase $i$, and we obtain $n_i \cdot B_i = n_0 \cdot B_0 = n \cdot d^3$.
	Therefore, the number of edges added in any phase $i$ does not exceed $n \cdot d^3$, and since all added edges are removed again at the end of each phase, the lemma statement follows.
\end{proof}

\begin{remark}
In the analysis, we implicitly assumed that edges incident on nodes are always added only once.
It could, however, be the case that some node is ``unlucky'' and many of its neighbors add a copy of the same edge many times.
This misfortune could potentially result in adding $n^\eps$ copies of the same (virtual) edge, which could, in turn, overload the memory per machine constraint on the machines containing these unlucky nodes.
For the sake of simplicity, we decided to leave this problem to the shuffling algorithm of the underlying MPC framework that can, for example, load the nodes onto the machines greedily after each communication step.
Since the total memory constraint is satisfied, this is always feasible.
Alternatively, the shuffling algorithm could simply drop duplicate messages.
\end{remark}

\newpage
\section{Open Questions}
 In this paper, we introduced a variant of the $\mpc$ model in which the standard assumption of $S=\widetilde{\Omega}(n)$ memory per machine is removed.

\paragraph{General Graphs}
We showed that in the case of the $\mis$ problem on trees this assumption is not necessary: Restricting the memory to $n^{\eps}$ per machine still allows an $\bigO(\log^2 \log n)$-round algorithm.
The first intriguing open problem follows. 
\begin{center}
	\begin{minipage}{0.95\textwidth}
		\begin{enumerate}
			\item[P1] \emph{Devise a low-memory $\mpc$ algorithm that finds an $\mis$ in general graphs in time $\poly \log \log n$.}
		\end{enumerate}
	\end{minipage}
\end{center}

\paragraph{Other Fundamental Graph Problems}
As an $\mis$ of the line graph\footnote{A line graph is a graph with a node for every edge in the input graph, and an edge between two nodes if the corresponding edges are incident.} corresponds to a Maximal Matching in the original graph, an $\mis$ algorithm usually directly gives rise to a Maximal Matching algorithm. 
In the $\mpc$ framework, however, it might not even be possible to store the line graph, which seems to complicate the simulation of the $\mis$ algorithm on the line graph.
 hard. 
the linear
, i.e., a graph where the edges of the input graph correspond to nodes in the line graph, in memory.
To the best of our knowledge, the field of Maximal Matching in the $\mpc$ world is wide open.
Naturally, there are many other standard graph problems some of which we list here.
\begin{center}
	\begin{minipage}{0.95\textwidth}
		\begin{enumerate}
			\item[P2] \emph{Devise an efficient low-memory $\mpc$ algorithm for the Maximal Matching problem.}
			\item[P4] \emph{Devise an efficient low memory $\mpc$ algorithm for the $(\Delta + 1)$-coloring problem.}
			\item[P5] \emph{Devise an efficient low memory $\mpc$ algorithm for the $(2\Delta - 1)$-edge-coloring problem.}
			\item[P4] \emph{Devise an efficient low memory $\mpc$ algorithm that finds an $\bigO(1)$-coloring in time $\bigO( \log d )$ in trees.}
		\end{enumerate}
	\end{minipage}
\end{center}

 \bibliographystyle{alpha}
 \bibliography{ref}

\end{document}